\documentclass[12pt]{iopart}

\usepackage{iopams,setstack}
\usepackage[T1]{fontenc}
\usepackage[english]{babel}
\usepackage{graphicx}

\newcommand{\N}{\ensuremath{\mathbb{N}}}

\newtheorem{theorem}{Theorem}

\newtheorem{proposition}[theorem]{Proposition}

\newenvironment{proof}[1][Proof]{\begin{trivlist}
\item[\hskip \labelsep {\bfseries #1}]}{\end{trivlist}}

\newcommand{\qed}{$\quad \Box$}

\begin{document}

\title{Information entropy of Gegenbauer polynomials of integer
parameter}

\author{Julio I. de Vicente$^1$, Silvia Gandy$^2$ and Jorge
S\'anchez-Ruiz$^{1,3}$}

\address{$^1$ Departamento de Matem\'aticas, Universidad Carlos III de
Madrid, \\ Avda.\ de la Universidad 30, E-28911 Legan\'es, Madrid,
Spain}

\address{$^2$ Fakult\"at f\"ur Mathematik, Technische Universit\"at M\"unchen,
\\ Boltzmannstr.\ 3, D-85747 Garching b.\ M\"unchen, Germany}

\address{$^3$ Instituto Carlos I de F\'{\i}sica Te\'orica y
Computacional, Universidad de Granada, E-18071 Granada, Spain}

\eads{\mailto{jdvicent@math.uc3m.es}, \mailto{gandy@in.tum.de},
\mailto{jsanchez@math.uc3m.es}}

\begin{abstract}
The position and momentum information entropies of $D$-dimensional
quantum systems with central potentials, such as the isotropic
harmonic oscillator and the hydrogen atom, depend on the entropies
of the (hyper)spherical harmonics. In turn, these entropies are
expressed in terms of the entropies of the Gegenbauer
(ultraspherical) polynomials $C_n^{(\lambda)}(x)$, the parameter
$\lambda$ being either an integer or a half-integer number. Up to
now, however, the exact analytical expression of the entropy of
Gegenbauer polynomials of arbitrary degree $n$ has only been
obtained for the particular values of the parameter $\lambda=0,1,2$.
Here we present a novel approach to the evaluation of the
information entropy of Gegenbauer polynomials, which makes use of
trigonometric representations for these polynomials and complex
integration techniques. Using this method, we are able to find the
analytical expression of the entropy for arbitrary values of both
$n$ and $\lambda\in\mathbb{N}$.
\end{abstract}



\submitto{\JPA}

\maketitle

\section{Introduction} \label{section1}

According to Shannon's information theory \cite{Sha}, the only
rigorous measure of the uncertainty or lack of information
associated to a continuous random variable $X$ with density
function $\rho(x)$, $x \in \mathbb{R}^D$, is the entropy
\begin{equation}\label{entropy}
H(X) = - \int \rho(x) \log \rho(x) \, \rmd x \;.
\end{equation}
In particular, when $\rho (x)$ is the single-particle probability
density for position of a quantum system, $H(X)$ is the only
rigorous measure of the uncertainty in the localization of the
particle in position space. The momentum entropy $H(P)$ can be
defined likewise from the single-particle density of momentum
$\gamma(p)$. In the simplest case of a single-particle system
described in position space by the wave function $\psi(x)$, we
have that $\rho(x) = \vert \psi(x) \vert^2$ and $\gamma(p) = \vert
\phi(p) \vert^2$, where the wave function in momentum space
$\phi(p)$ is the Fourier transform of $\psi(x)$. The sharp
inequality \cite{Bec,Bia}
\begin{equation}\label{bbm}
H(X) + H(P) \geq  D \left( 1 + \log \pi \right)
\end{equation}
places a nontrivial lower bound on the sum of the uncertainties in
position and momentum, so it provides a quantitative formulation of
the position-momentum uncertainty principle. Using the variational
inequality that relates information entropy and standard deviation
for an arbitrary $D$-dimensional random variable \cite{Sha,Bia},
\begin{equation}
H(A) \leq \frac{D}{2} \left( 1 + \log \frac{2\pi (\Delta A)^2}{D}
\, \right),
\end{equation}
the entropic uncertainty relation (\ref{bbm}) leads to the
well-known Heisenberg uncertainty relation
\begin{equation}
\Delta X \Delta P \geq \frac{D}{2}\;,
\end{equation}
which proves the former to be stronger than the latter.

For many important quantum systems, such as $D$-dimensional harmonic
oscillator and hydrogen atom, the calculation of position and
momentum information entropies involves the evaluation of integrals
of the form
\begin{equation} \label{eop}
E(p_n) = - \int_a^b \big(p_n(x)\big)^2 \log \big(p_n(x)\big)^2 \,
\omega(x) \, \rmd x \,,
\end{equation}
where $\{ p_n(x) \}$ denotes a polynomial sequence (${\rm deg} \,
p_n(x) = n$) orthogonal on $[a,b] \subseteq \mathbb{R}$ with
respect to the weight function $\omega(x)$. During the last decade
there has been an intense activity in the study of these
integrals, motivated not only by their relevance to quantum
physics but also by their close relationship to other interesting
mathematical objects, such as the $L^p$-norms or the logarithmic
potentials of the polynomials $p_n(x)$. A survey on the
state-of-the-art in this field up to year 2001 can be found in
\cite{review}.

The calculation of the entropic integrals $E(p_n)$ is generally a
very difficult task, and in most cases only asymptotic results for
large values of $n$ are known \cite{review}. In fact, since all the
zeros of $p_n$ are simple and belong to $(a,b)$, when $n$ is not
very small even a numerical computation of $E(p_n)$ poses serious
difficulties due to the strongly oscillatory behaviour of the
integrand in (\ref{eop}). In this respect it is worth mentioning
Ref.\ \cite{Buy04}, which presents an efficient algorithm for the
numerical evaluation of $E(p_n)$ in the case when the interval
$(a,b)$ is finite.

Closed analytical formulas for $E(p_n)$ are only known for a few
particular cases of the Gegenbauer or ultraspherical polynomials
$C_n^{(\lambda)}$. We recall that these polynomials are defined as
(see, e.g., \cite[Sec.\ 4.7]{Sze})
\begin{equation}
C_n^{(\lambda)} (x) = \frac{(2 \lambda)_n}{(\lambda +
\frac{1}{2})_n} P_n^{\left(
\lambda-\frac{1}{2},\lambda-\frac{1}{2} \right)}(x) \;,
\end{equation}
where $(a)_n = \Gamma(a+n) / \Gamma(a)$ denotes the Pochhammer
symbol and $P_n^{(\alpha,\beta)}(x)$ are Jacobi polynomials,
\begin{equation}
P_n^{(\alpha,\beta)}(x) = \frac{(\alpha +1)_n}{n!} \hspace{0.1cm}
{_2F_1} \bigg( \begin{array}{c} -n, \, n+\alpha+\beta+1 \\ \alpha +1
\end{array} \bigg| \frac{1-x}{2} \bigg) \,.
\end{equation}
For $\lambda > - \frac{1}{2}$, Gegenbauer polynomials form an
orthogonal sequence on the interval $[-1,1]$ with respect to the
weight function $w_{\lambda}(x) = (1-x^2)^{\lambda - \frac{1}{2}}$,
\begin{equation} \label{gegengauer}
\int_{-1}^1 C_n^{(\lambda)} (x) C_m^{(\lambda)} (x) (1-x^2)^{\lambda
- \frac{1}{2}} \, \rmd x = \frac{2^{1-2 \lambda} \, \pi \, \Gamma
(n+2 \lambda)}{(n+ \lambda) \, n! \, [ \Gamma (\lambda) ]^2} \,
\delta _{n,m} \;.
\end{equation}
The information entropies of Gegenbauer polynomials, on which we
focus in the present paper, are thus given by
\begin{equation}\label{entrgeg}
E(C_n^{(\lambda)})=-\int_{-1}^1 \big(C_n^{(\lambda)}(x)\big)^2
\log\big(C_n^{(\lambda)}(x)\big)^2(1-x^2)^{\lambda-\frac{1}{2}} \,
\rmd x \;.
\end{equation}

The integrals $E(C_n^{(\lambda)})$ are especially relevant in the
case when $\lambda$ is a non-negative integer or half-integer
number, due to the relationship between the corresponding Gegenbauer
polynomials and (hyper)spherical harmonics. As a consequence, these
integrals appear in the calculation of the angular component of
information entropies in both position and momentum space for any
$D$-dimensional ($D\geq 2$) quantum-mechanical system with a central
potential, such as the isotropic harmonic oscillator or the hydrogen
atom (radially symmetric Coulomb potential)
\cite{review,Yan94,Deh97,Yan99}. They also control the radial
component of the information entropy in momentum space for the
$D$-dimensional hydrogen atom \cite{review,Yan94,Deh97}.

Instead of using the standard definition of Gegenbauer
polynomials, it is often more convenient to work with the
polynomials
\begin{equation}
\widehat{C}_n^{(\lambda)} (x) = \left( \frac{(n+\lambda) \,
n!}{\lambda \, (2\lambda)_n} \right)^{\! \! \frac{1}{2}}
C_n^{(\lambda)} (x) \,,
\end{equation}
which are orthonormal on $[-1,1]$ with respect to the probability
density
\begin{equation}
\widehat{w}_{\lambda}(x) = \frac{\Gamma (\lambda +1)}{\sqrt{\pi}
\, \Gamma (\lambda + \frac{1}{2})} (1-x^2)^{\lambda-\frac{1}{2}}
\;.
\end{equation}
The corresponding entropies,
\begin{equation} \label{egop2}
E(\widehat{C}_n^{(\lambda)}) = - \int_{-1}^1
[\widehat{C}_n^{(\lambda)}(x)]^2 \log \,
[\widehat{C}_n^{(\lambda)}(x)]^2 \, \widehat{w}_{\lambda}(x) \, \rmd
x \;,
\end{equation}
are related to $E(C_n^{(\lambda)})$ by the formula
\begin{equation} \label{eogeonrel}
E(\widehat{C}_n^{(\lambda)}) = \log \left( \frac{\lambda \,
(2\lambda)_n}{(n+\lambda) \, n!} \right) + \frac{\Gamma (\lambda)
(n+\lambda) \, n!}{\sqrt{\pi} \, \Gamma \big( \lambda +
\frac{1}{2} \big) (2\lambda)_n} E(C_n^{(\lambda)}) \;,
\end{equation}
which readily follows from the previous definitions by taking into
account the orthogonality relation (\ref{gegengauer}).

The simplest particular cases of Gegenbauer polynomials are the
Chebyshev polynomials of the first and second kind,
\begin{equation}
T_n(x) = \lim_{\lambda \rightarrow 0} \frac{n!}{(2 \lambda)_n}\,
C_n^{(\lambda)}(x)\;, \quad U_n(x) = C_n^{(1)}(x)\;.
\end{equation}
For both of these families, information entropies can be computed in
closed analytical form, the results being \cite{Yan94,Deh97}
\begin{eqnarray}\label{echebyshev1}
E(\widehat{T}_n) = \cases{0 & if $n=0$\,, \\
  \log2-1 & if $n\neq0$\,, \\} \\ \label{echebyshev2}
E(\widehat{U}_n)=-\frac{n}{n+1}\;.
\end{eqnarray}
In the $\lambda = 2$ case, it was first proved in \cite{Buyarov:97}
that
\begin{equation} \label{l2entropy}
\fl E(\widehat{C}_n^{(2)}) = - \log \left( \frac{3(n+1)}{n+3}
\right) - \frac{n(n^2+2n-1)}{(n+1)(n+2)(n+3)} -
\frac{2}{\sqrt{(n+1)^3 (n+3)^3}}
\frac{T_{n+2}^{'''}(\xi)}{T_{n+2}^{''}(\xi)} \,,
\end{equation}
where
\begin{equation}
\xi = \frac{n+2}{\sqrt{(n+1)(n+3)}} \,,
\end{equation}
and this result was later simplified to \cite{Buy00}
\begin{equation} \label{gegl2entropy}
\fl E(\widehat{C}_n^{(2)}) = - \log \left( \frac{3(n+1)}{n+3}
\right) - \frac{n^3-5n^2-29n-27}{(n+1)(n+2)(n+3)} - \frac{1}{n+2}
\left( \frac{n+3}{n+1} \right)^{n+2} \,.
\end{equation}

In the same work \cite{Buy00}, it was also obtained the following
generalization of (\ref{l2entropy}) to arbitrary integer values of
the parameter, $\lambda = l \in \mathbb{N}$:
\begin{equation}\label{buyarov}
E(\widehat{C}_n^{(l)}) = -s_{nl}-r_{nl} \sum_{j=1}^{2l-2}
(1-\xi^2_j) \frac{H(\xi_j)}{P'(\xi_j)}
\frac{\widehat{C}_{n-1}^{(l+1)}(\xi_j)}{\widehat{C}_n^{(l)}(\xi_j)}
\,,
\end{equation}
where $s_{nl}$ and $r_{nl}$ are known constants depending only on
$n$ and $l$, the auxiliary polynomials $P$ and $H$ are defined
from the sequence $\{ P_k \}$ (${\rm deg} \, P_k = k$) generated
by the recurrence relation
\begin{equation}
\fl P_{k+1}(x) = (2l-2k-3) x P_k (x) - (n+k+1)(n+2l-k-1) (1-x^2)
P_{k-1}(x)
\end{equation}
from the initial values $P_{-1}(x)=0, P_0(x)=1$ through the
formulas
\begin{equation}
P(x) = P_{2l-2} (x) \;, \quad H(x) = \sum_{s=0}^{2l-2} (-1)^s
P_{s-1}(x) P_{2l-s-3}(x) \;,
\end{equation}
and $\xi_j$ ($j = 1, 2, \ldots, 2l-2$) denote the zeros of $P$.
The explicit expression of the polynomial $P$ was later found to
be \cite{San03}
\begin{equation} \label{pxexplicit}
\fl P(x) = \frac{(-1)^{l-1} (n+2l-1)!}{(n+l) n!} \sum_{\mu=0}^{l-1}
\frac{(1-l)_{\mu} (l)_{\mu} (1/2)_{\mu}}{(1-n-l)_{\mu} (1+n+l)_{\mu}
\, \mu!} (1-x^2)^{l-1-\mu} \,.
\end{equation}

Regretfully, (\ref{buyarov}) is not easy to use in practice.
Furthermore, it is not a completely analytical formula save for
small values of $l$ since, as we readily see from
(\ref{pxexplicit}), the zeros $\xi_j$ of $P$ have to be determined
numerically when $l\geq6$\footnote{Likewise, the general expression
of $E(C_n^{(\lambda)})$ given in \cite{Yan99} is not completely
analytical save for small values of $n$, since it is expressed in
terms of the zeros of $C_n^{(\lambda)}(x)$.}.

As first pointed out in \cite{Sanpreprint}, the entropy of Chebyshev
polynomials of the first and second kind can be easily computed by
direct calculation of the corresponding integrals by using the
well-known trigonometric representations
\begin{eqnarray}\label{trigt}
T_n(\cos\theta)=\cos n\theta\;,\quad
U_n(\cos\theta)=\frac{\sin(n+1)\theta}{\sin\theta} \;,
\end{eqnarray}
with $x=\cos\theta$. Motivated by this observation, in the present
paper we aim at evaluating the entropic integral
$E(C_n^{(\lambda)})$ for general values of the parameter $\lambda$
using representations of the same kind for the Gegenbauer
polynomials.

We begin by collecting, in Section \ref{section2}, the trigonometric
representations of Gegenbauer polynomials that will be used later
on. Our approach is developed in Section \ref{section3}, where we
show that it enables us to find completely analytical expressions
for $E(C_n^{(\lambda)})$, in terms of finite sums, whenever
$\lambda\in\mathbb{N}$. The new results obtained for the information
entropy of Gegenbauer polynomials of integer parameter are
summarized in Section \ref{section4}. Finally, in Section
\ref{section5} some concluding remarks are given and several open
problems are pointed out.

\section{Trigonometric representations for Gegenbauer
polynomials} \label{section2}

The most widely known trigonometric representation of the
Gegenbauer polynomials is (see e.g.\ \cite[p.\ 302]{AAR})
\begin{equation}\label{standard}
C_n^{(\lambda)}(\cos\theta) = \sum_{m=0}^n
d_{m,n}^{(\lambda)}{\rme}^{\rmi (n-2m)\theta} = \sum_{m=0}^n
d_{m,n}^{(\lambda)}\cos(n-2m)\theta\,,
\end{equation}
where
\begin{equation}\label{d}
d_{m,n}^{(\lambda)}=\frac{(\lambda)_m(\lambda)_{n-m}}{m!(n-m)!}\;.
\end{equation}
Another representation, due to Szeg\"o \cite{Sze,Sze2}, is
\begin{equation}\label{szego}
C_n^{(\lambda)}\left(\cos\theta\right) =
\frac{c_n^{(\lambda)}}{\left(\sin\theta\right)^{2\lambda-1}}\sum_{\nu=0}^\infty
\alpha_{\nu,n}^{(\lambda)}\sin
(n+2\nu+1)\theta\,,\quad\lambda>0\,,\,\lambda\notin\mathbb{N}\,,
\end{equation}
where
\begin{equation}\label{calfa}
c_n^{(\lambda)}=\frac{2^{2-2\lambda}\Gamma(n+2\lambda)}{\Gamma(\lambda)\Gamma(n+\lambda+1)}
\;,\quad \alpha_{\nu,n}^{(\lambda)} =
\frac{(1-\lambda)_{\nu}(n+1)_\nu}{\nu!(n+\lambda+1)_\nu} \;.
\end{equation}
At first sight, this representation seems to be less useful than the
previous one, because it contains infinitely many terms. Moreover,
it is supposed not to hold when $\lambda\in\mathbb{N}$. However, it
is not difficult to prove that the validity of (\ref{szego}) extends
to the case when $\lambda$ is a positive integer.

\begin{proposition}\label{propszego}
The Szeg\"o representation (\ref{szego}) holds true when
$\lambda\in\mathbb{N}$. In this case, it reads
\begin{equation}\label{szegotruncada}
C_n^{(\lambda)}\left(\cos\theta\right) =
\frac{c_n^{(\lambda)}}{\left(\sin\theta\right)^{2\lambda-1}}
\sum_{\nu=0}^{\lambda-1} \alpha_{\nu,n}^{(\lambda)}\sin
(n+2\nu+1)\theta\,.
\end{equation}
\end{proposition}

\begin{proof}
If $\lambda\in\mathbb{N}$ then $\alpha_{\nu,n}^{(\lambda)}=0$ when
$\nu\geq\lambda$, so that (\ref{szego}) reduces to
(\ref{szegotruncada}). We will prove this equality by induction on
$\lambda$. When $\lambda=1$ (\ref{szegotruncada}) is obviously true
since it reduces to the second equation in (\ref{trigt}), the
well-known trigonometric representation for the Chebyshev
polynomials of the second kind. Now, assume that the result holds
for $\lambda=m-1$ ($m\in\mathbb{N}$). We take advantage of the
following recurrence relation for the Gegenbauer polynomials
\cite[Eq.\ (4.7.27)]{Sze},
\begin{equation}
nC_n^{(\lambda)}(x) = (2\lambda+n-1)xC_{n-1}^{(\lambda)}(x) -
2\lambda(1-x^2)C_{n-2}^{(\lambda+1)}(x),
\end{equation}
which in trigonometric form ($x=\cos\theta$) can be restated as
\begin{equation}\label{gegrecursion}
\fl C_n^{(\lambda)}(\cos\theta) =
\frac{1}{2(\lambda-1)\sin^2\theta}\left[(2\lambda+n-1)\cos\theta
C_{n+1}^{(\lambda-1)}(\cos\theta) -
(n+2)C_{n+2}^{(\lambda-1)}(\cos\theta)\right].
\end{equation}
Using this formula for $\lambda=m$ and substituting
(\ref{szegotruncada}) on the right-hand-side we arrive at
\begin{eqnarray}
\fl C_n^{(m)}(\cos\theta) & = &
\frac{2c_n^{(m)}}{\left(\sin\theta\right)^{2m-1}} \left[\cos\theta
\sum_{\nu=0}^{m-2} \alpha_{\nu,n+1}^{(m-1)}\sin(n+2\nu+2)\theta
\right. \nonumber \\ \fl & & \left. -\frac{n+2}{n+m+1}
\sum_{\nu=0}^{m-2}\alpha_{\nu,n+2}^{(m-1)}\sin(n+2\nu+3)\theta
\right] \nonumber \\ \fl & = &
\frac{c_n^{(m)}}{\left(\sin\theta\right)^{2m-1}}
\left[\sum_{\nu=0}^{m-2}\alpha_{\nu,n+1}^{(m-1)}\sin(n+2\nu+3)\theta
+ \sum_{\nu=0}^{m-2}\alpha_{\nu,n+1}^{(m-1)}\sin(n+2\nu+1)\theta
\right. \nonumber \\ \fl & & \left. -\frac{2(n+2)}{n+m+1}
\sum_{\nu=0}^{m-2} \alpha_{\nu,n+2}^{(m-1)} \sin(n+2\nu+3)\theta \right] \nonumber \\
\fl & = & \frac{c_n^{(m)}}{\left(\sin\theta\right)^{2m-1}}
\sum_{\nu=0}^{m-1}\left(\alpha_{\nu-1,n+1}^{(m-1)}+\alpha_{\nu,n+1}^{(m-1)}
-\frac{2(n+2)}{n+m+1}\alpha_{\nu-1,n+2}^{(m-1)}\right) \nonumber \\
\fl & & \times \sin(n+2\nu+1)\theta\;,
\end{eqnarray}
where in the last step we have used that
$\alpha_{m-1,n+1}^{(m-1)}=0$ and $\alpha_{\nu,n+1}^{(m-1)}=0$
whenever $\nu<0$. A straightforward calculation shows that
\begin{equation}
\alpha_{\nu-1,n+1}^{(m-1)}+\alpha_{\nu,n+1}^{(m-1)} -
\frac{2(n+2)}{n+m+1}\alpha_{\nu-1,n+2}^{(m-1)} =
\alpha_{\nu,n}^{(m)}\;,
\end{equation}
and (\ref{szegotruncada}) is thus proved to hold also for
$\lambda=m$. \qed
\end{proof}

The fact that the sum in (\ref{szego}) terminates after a finite
number of terms when $\lambda\in\mathbb{N}$ suggests that
Szeg\"o's representation may be useful to evaluate the entropy of
Gegenbauer polynomials of integer parameter. Accordingly, in what
follows we shall assume that $\lambda\in\mathbb{N}$ unless
otherwise indicated.

\section{Evaluation of the entropic integral} \label{section3}

With the change of variable $x=\cos\theta$, the integral
(\ref{entrgeg}) takes the form
\begin{equation} \label{entrotrigo}
E(C_n^{(\lambda)}) = -\int_{0}^\pi \big(
            C_n^{(\lambda)}\left(\cos\theta\right)\big)^2
            \log\big(C_n^{(\lambda)}\left(\cos\theta\right)\big)^2
            \sin^{2\lambda} \theta \, \rmd \theta\;.
\end{equation}
Using Szeg\"o's representation (\ref{szegotruncada}) for one of the
two Gegenbauer polynomials in
$\big(C_n^{(\lambda)}(\cos\theta)\big)^2$, (\ref{entrotrigo}) can be
rewritten as
\begin{equation} \label{entrotrigo2}
E(C_n^{(\lambda)}) = -\frac{1}{2} \, c_n^{(\lambda)}
\sum_{\nu=0}^{\lambda -1} \alpha_{\nu,n}^{(\lambda)}
\left(J_{\nu,n}^{(\lambda)}-J_{\nu+1,n}^{(\lambda)}\right),
\end{equation}
where
\begin{equation}
J_{\nu,n}^{(\lambda)} := \int_0^\pi C_n^{(\lambda)}
(\cos\theta)\cos(n+2\nu)\theta\log\big(C_n^{(\lambda)}(\cos\theta)\big)^2
\rmd \theta \;.
\end{equation}
Now, using the standard representation (\ref{standard}) we have that
\begin{equation}
\fl C_n^{(\lambda)}(\cos\theta)\cos(n+2\nu)\theta = \frac{1}{2}
\sum_{m=0}^n d_{m,n}^{(\lambda)}\cos2(m+\nu)\theta + \frac{1}{2}
\sum_{m=0}^nd_{m,n}^{(\lambda)}\cos2(n-m+\nu)\theta \,.
\end{equation}
Taking into account the symmetry property $d_{m,n}^{(\lambda)} =
d_{n-m,n}^{(\lambda)}$, which readily follows from the explicit
expression of the coefficients $d_{m,n}^{(\lambda)}$, the previous
equation simplifies to
\begin{equation}
C_n^{(\lambda)}(\cos\theta)\cos(n+2\nu)\theta =\sum_{m=0}^n
d_{m,n}^{(\lambda)}\cos2(m+\nu)\theta\;,
\end{equation}
so that
\begin{equation} \label{intjota}
J_{\nu,n}^{(\lambda)}=\sum_{m=0}^n d_{m,n}^{(\lambda)}\int_0^\pi
\cos2(m+\nu)\theta\log\big(C_n^{(\lambda)}(\cos\theta)\big)^2 \rmd
\theta\;.
\end{equation}
Defining the integrals
\begin{equation}\label{integral}
I_{m,n}^{(\lambda)} := \int_0^\pi \cos(2m\theta) \log\big(
C_n^{(\lambda)}(\cos\theta)\big)^2 \rmd \theta\;,
\end{equation}
from (\ref{entrotrigo2}) and (\ref{intjota}) we find that
$E(C_n^{(\lambda)})$ is given by
\begin{equation} \label{entrotrigo3}
E(C_n^{(\lambda)})=-\frac{1}{2}\,c_n^{(\lambda)}
\sum_{\nu=0}^{\lambda - 1}\alpha_{\nu,n}^{(\lambda)} \sum_{m=0}^n
d_{m,n}^{(\lambda)} \left( I_{\nu+m,n}^{(\lambda)} - I_{\nu+1
+m,n}^{(\lambda)} \right).
\end{equation}

\setcounter{theorem}{0}

An alternative expression for the entropic integral
$E(C_n^{(\lambda)})$ which turns out to be more convenient in
practice can be obtained by noticing that
\begin{eqnarray}
\fl E(C_n^{(\lambda)}) & = -\frac{1}{2}\,c_n^{(\lambda)}
\sum_{\nu=0}^{\lambda - 1}\alpha_{\nu,n}^{(\lambda)}
\left\{\sum_{m=0}^n d_{m,n}^{(\lambda)}
 I_{\nu+m,n}^{(\lambda)} - \sum_{m=1}^{n+1} d_{m-1,n}^{(\lambda)}
 I_{\nu +m,n}^{(\lambda)} \right\}\nonumber\\
\fl & = -\frac{1}{2} \, c_n^{(\lambda)}\sum_{\nu=0}^{\lambda -
1}\alpha_{\nu,n}^{(\lambda)} \left\{\sum_{m=1}^{n} \left(
d_{m,n}^{(\lambda)}- d_{m-1,n}^{(\lambda)} \right)
I_{\nu+m,n}^{(\lambda)} + d_{0,n}^{(\lambda)}I_{\nu,n}^{(\lambda)}-
d_{n,n}^{(\lambda)}I_{\nu+n+1,n}^{(\lambda)} \right\}.\label{31}
\end{eqnarray}
According to (\ref{d}), $d_{m,n}^{(\lambda)}=0$ when $-\lambda<m<0$
or $n<m<n+\lambda$. Let us restrict initially to the case when
$\lambda\neq1$, so that $d_{-1,n}^{(\lambda)}=0$ and
$d_{n+1,n}^{(\lambda)}=0$. This allows us to write the previous
formula in the more compact form
\begin{eqnarray}
E(C_n^{(\lambda)}) & = -\frac{1}{2} \,
c_n^{(\lambda)}\sum_{\nu=0}^{\lambda -
1}\sum_{m=0}^{n+1}\alpha_{\nu,n}^{(\lambda)} \left(
d_{m,n}^{(\lambda)}-
d_{m-1,n}^{(\lambda)} \right) I_{\nu+m,n}^{(\lambda)}\nonumber\\
& = -\frac{1}{2} \, c_n^{(\lambda)}\sum_{\nu=0}^{\lambda -
1}\sum_{m=\nu}^{n+1+\nu}\alpha_{\nu,n}^{(\lambda)} \left(
d_{m-\nu,n}^{(\lambda)}- d_{m-\nu-1,n}^{(\lambda)} \right)
I_{m,n}^{(\lambda)}\;.
\end{eqnarray}
Using again that $d_{m,n}^{(\lambda)}=0$ when $-\lambda<m<0$ as well
as when $n<m<n+\lambda$, we can extend the lower and upper limits in
the inner summation to 1 and $n+\lambda-1$, respectively, provided
that the terms $m=0$ and $m=n+\lambda$ are treated separately. Thus
we find that
\begin{equation}\label{suma}
\fl E(C_n^{(\lambda)}) =-\frac{1}{2} \,
c_n^{(\lambda)}\left(\alpha_{0,n}^{(\lambda)}d_{0,n}^{(\lambda)}I_{0,n}^{(\lambda)}
-\alpha_{\lambda-1,n}^{(\lambda)}d_{n,n}^{(\lambda)}I_{n+\lambda,n}^{(\lambda)}
+\sum_{m=1}^{n+\lambda-1}\beta_{m,n}^{(\lambda)}I_{m,n}^{(\lambda)}\right),
\end{equation}
where
\begin{equation}\label{beta}
\beta_{m,n}^{(\lambda)}=\sum_{\nu=0}^{\lambda -
1}\alpha_{\nu,n}^{(\lambda)} \left( d_{m-\nu,n}^{(\lambda)}-
d_{m-\nu-1,n}^{(\lambda)} \right).
\end{equation}
It can be seen that (\ref{suma}) also holds when $\lambda=1$ by
noting that in this case its right-hand side coincides with that of
(\ref{31}).

In order to apply (\ref{suma}), we need to evaluate the integrals
$I_{m,n}^{(\lambda)}$ with $0 \leq m \leq n+\lambda$. This goal can
be achieved by means of complex integration techniques, which enable
us to obtain the following result.

\begin{theorem}\label{t1}
For $\lambda\in\mathbb{N}$,
\begin{equation} \label{icero}
I_{0,n}^{(\lambda)} = 2\pi \log\left( \frac{(\lambda)_n}{n!}\right)
\end{equation}
and, when $m\geq1$,
\begin{equation} \label{formulat1}
\fl I_{m,n}^{(\lambda)} = \frac{(2\lambda-1)\pi}{m} +
\frac{\pi}{(2m)!} \, \frac{d^{2m}}{dz^{2m}}\left. \left(
\log\sum_{\nu=0}^{\lambda-1} \alpha_{\nu,n}^{(\lambda)} \big(
z^{2n+2\lambda+2\nu}-z^{2\lambda-2\nu-2} \big) \right)\right|_{z=0}.
\end{equation}
\end{theorem}

\begin{proof}
Taking into account that
$C_n^{(\lambda)}(-x)=(-1)^nC_n^{(\lambda)}(x)$, (\ref{integral}) can
be written as
\begin{eqnarray} \label{imntrig}
I_{m,n}^{(\lambda)} & = \frac{1}{2}\int_0^{2\pi} \cos(2m\theta)
\log\big|C_n^{(\lambda)}(\cos\theta)\big|^2 \rmd\theta \nonumber \\
& = \int_0^{2\pi} \cos(2m\theta)
\log\big|C_n^{(\lambda)}(\cos\theta)\big| \, \rmd\theta \nonumber \\
& = \int_0^{2\pi} \cos(2m\theta) \log\big|{\rme}^{\rmi
n\theta}C_n^{(\lambda)}(\cos\theta)\big| \, \rmd\theta \;,
\end{eqnarray}
where in the last step the factor ${\rme}^{\rmi n\theta}$ has been
introduced for later convenience. Using the Szeg\"{o} representation
(\ref{szegotruncada}) for the Gegenbauer polynomial
$C_n^{(\lambda)}(\cos\theta)$, the previous equation reads
\begin{equation}
I_{m,n}^{(\lambda)}=\int_0^{2\pi} \cos(2m\theta) \log \left|
\frac{c_n^{(\lambda)}{\rme}^{\rmi n\theta}}{\sin^{2\lambda-1}\theta}
\sum_{\nu=0}^{\lambda-1}\alpha_{\nu,n}^{(\lambda)}\sin(n+2\nu+1)\theta
\right| \rmd\theta\;.
\end{equation}
We will compute the integral
\begin{equation}
\mathcal{I}_{m,n}^{(\lambda)}=\int_0^{2\pi} \cos(2m\theta) \log
\left(\frac{c_n^{(\lambda)}{\rme}^{\rmi n
\theta}}{\sin^{2\lambda-1}\theta} \sum_{\nu=0}^{\lambda-1}
\alpha_{\nu,n}^{(\lambda)}\sin(n+2\nu+1)\theta \right) \rmd\theta\;,
\end{equation}
whose real part equals $I_{m,n}^{(\lambda)}$. Introducing the change
of variable $z=\exp(\rmi\theta)$ we arrive at
\begin{equation}\label{i}
\mathcal{I}_{m,n}^{(\lambda)}=\frac{1}{2\rmi}\oint_{|z|=1}
\frac{z^{4m}+1}{z^{2m+1}} \log q(z) \rmd z \;,
\end{equation}
where
\begin{equation}\label{q}
\fl q(z) = z^n \, C_n^{(\lambda)}\left(\frac{z+z^{-1}}{2}\right) =
c_n^{(\lambda)}2^{2\lambda-2}(-1)^{\lambda}
\frac{\sum_{\nu=0}^{\lambda-1}\alpha_{\nu,n}^{(\lambda)} \big(
z^{2n+2\lambda+2\nu} - z^{2\lambda-2\nu-2}
\big)}{(1-z^2)^{2\lambda-1}}\;.
\end{equation}
The singularities of the integrand are $z=0$, which is a pole of
order $2m+1$, and all the zeros of $q(z)$, which are branch points.
If $\{x_{n,j}\}_{j=1}^n$ denote the zeros of $C_n^{(\lambda)}(x)$,
which are known to be simple, real and located in $(-1,1)$, then the
zeros $\{z_{n,j}\}_{j=1}^{2n}$ of the function $q(z)$ are
\begin{equation}\label{zeros}
z_{n,j+\frac{n}{2} \mp \frac{n}{2}}=\exp(\rmi\arccos
x_{n,j})=x_{n,j}\pm \rmi\sqrt{1-x_{n,j}^2} \;, \quad j=1,2,\ldots,n
\;.
\end{equation}
This means that the $\{z_{n,j}\}_{j=1}^{2n}$ are all located on the
unit circle, which can also be seen from the fact that
$z=\exp(\rmi\arccos x)$ maps $(-1,1)$ onto the unit circle.
Therefore, the integrand of (\ref{i}) has $2n$ branch points located
on the contour of integration. To avoid this difficulty we consider
the same integral along the slightly different contour $\Gamma$ (see
Figure \ref{figure1}), which is also closed. Notice that the
logarithmic branches can be chosen to go from the branch points to
the exterior of the unit disk, so that $\Gamma$ does not cross them.
Since the only singularity inside $\Gamma$ is $z=0$ we now have
\begin{equation}\label{residuo}
\oint_{\Gamma} \frac{z^{4m}+1}{z^{2m+1}} \log q(z) \rmd z = 2\pi
\rmi\, \mbox{Res} \left(\frac{z^{4m}+1}{z^{2m+1}} \log
q(z),z=0\right).
\end{equation}
\begin{figure}
\centering
\includegraphics{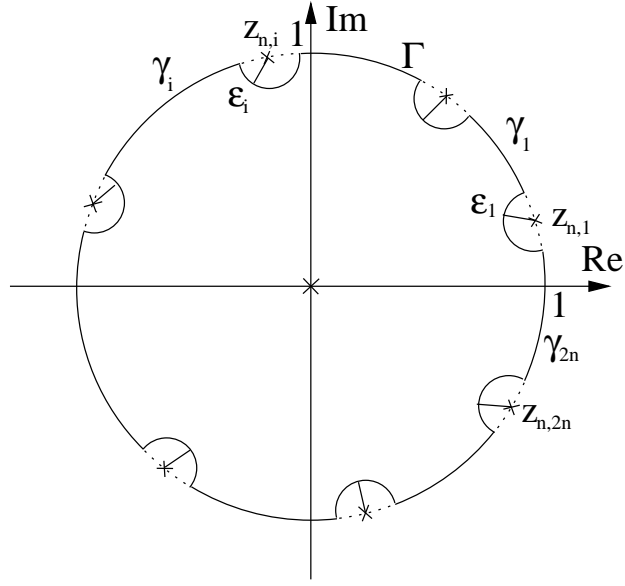}
\caption{Integration contour used to avoid the branch points on the
unit circle.} \label{figure1}
\end{figure}
The integral along $\Gamma$ can be decomposed as
\begin{equation}\label{integralGamma}
\fl \oint_\Gamma\frac{z^{4m}+1}{z^{2m+1}} \log q(z) \rmd z =
\sum_{j=1}^{2n}\left(\int_{\epsilon_j}\frac{z^{4m}+1}{z^{2m+1}} \log
q(z)\rmd z + \int_{\gamma_j}\frac{z^{4m}+1}{z^{2m+1}} \log q(z) \rmd
z \right),
\end{equation}
where $\epsilon_j$ denotes the arc of circumference of radius
$\varepsilon$ that surrounds the branch point $z_{n,j}$ and
$\gamma_j$ denotes the arc on the unit circle that connects
$\epsilon_j$ and $\epsilon_{j+1}$ ($\gamma_{2n}$ connects
$\epsilon_{2n}$ and $\epsilon_{1}$). Parameterizing
$z=z_{n,j}+\varepsilon {\rme}^{\rmi\theta}$ we find that
\begin{equation}
\fl \int_{\epsilon_j}\frac{z^{4m}+1}{z^{2m+1}} \log q(z) \rmd z =
\int_{\epsilon_j}\frac{(z_{n,j}+\varepsilon
{\rme}^{\rmi\theta})^{4m}+1}{(z_{n,j}+\varepsilon
{\rme}^{\rmi\theta})^{2m+1}} \log q \big( z_{n,j}+\varepsilon
{\rme}^{\rmi\theta} \big) \, \rmi\varepsilon
{\rme}^{\rmi\theta}d\theta\underset{\varepsilon\rightarrow0}{\longrightarrow}0\;,
\end{equation}
where we have used that $x\log x \to 0$ as $x \to 0$. Thus, taking
the limit $\varepsilon \to 0$ in (\ref{integralGamma}) we conclude
that
\begin{equation} \label{c1igualc2}
\oint_{|z|=1}\frac{z^{4m}+1}{z^{2m+1}} \log q(z) \rmd{z} =
\oint_\Gamma\frac{z^{4m}+1}{z^{2m+1}} \log q(z) \rmd{z} \;.
\end{equation}
Taking into account that the residue of a meromorphic function
$h(z)$ in a pole $z_0$ of order $2m+1$ is given by
\begin{equation}
\mbox{Res} \big( h(z), z=z_0 \big) = \left.
\frac{1}{(2m)!}\frac{d^{2m}}{dz^{2m}}\left((z-z_0)^{2m+1}
h(z)\right)\right|_{z=z_0} \,,
\end{equation}
use of (\ref{residuo}) and (\ref{c1igualc2}) into (\ref{i}) leads to
\begin{equation} \label{icompleja}
\mathcal{I}_{m,n}^{(\lambda)} = \left.
\frac{\pi}{(2m)!}\frac{d^{2m}}{dz^{2m}} \left[ (z^{4m}+1) \log
q(z) \right]\right|_{z=0}\,.
\end{equation}
In the case $m=0$, the previous equation reduces to
\begin{equation}
\mathcal{I}_{0,n}^{(\lambda)} = 2\pi\log q(0) =
2\pi\log\left(c_n^{(\lambda)}2^{2\lambda-2}
(-1)^{\lambda+1}\alpha_{\lambda-1,n}^{(\lambda)}\right),
\end{equation}
so that
\begin{equation}
I_{0,n}^{(\lambda)}=2\pi\log\left|c_n^{(\lambda)}2^{2\lambda-2}
\alpha_{\lambda-1,n}^{(\lambda)}\right|=2\pi \log\left(
\frac{(\lambda)_n}{n!}\right),
\end{equation}
which proves the first part of the theorem\footnote{This part can
also be proved using the mean value theorem for harmonic functions
(cf.\ \cite[Sec. VI]{Yan94}).}. On the other hand, if $m\geq1$ then
we readily see from (\ref{icompleja}) that
$\mathcal{I}_{m,n}^{(\lambda)} \in \mathbb{R}$, so
$I_{m,n}^{(\lambda)}=\mathcal{I}_{m,n}^{(\lambda)}$. Furthermore, in
this case the factor $(z^{4m}+1)$ in the right-hand side of
(\ref{icompleja}) can be omitted, since at $z=0$ its value equals
unity while all its derivatives do vanish. We thus find that
\begin{equation}
I_{m,n}^{(\lambda)} = \frac{\pi}{(2m)!}\frac{d^{2m}}{dz^{2m}}
\left. \left[ \log \left(
\frac{\sum_{\nu=0}^{\lambda-1}\alpha_{\nu,n}^{(\lambda)}
(z^{2n+2\lambda+2\nu}-z^{2\lambda-2\nu-2})}{(1-z^2)^{2\lambda-1}}
\right)\right]\right|_{z=0}\,,
\end{equation}
and (\ref{formulat1}) follows on noting that
\begin{equation}
\left.\frac{d^{2m}}{dz^{2m}}\left(\log(1-z^2)\right)\right|_{z=0}
= \frac{d^{2m}}{dz^{2m}}\left.
\left(-\sum_{k=1}^\infty\frac{z^{2k}}{k}\right)\right|_{z=0} =
-\frac{(2m)!}{m}
\end{equation}
when $m\neq0$. \qed
\end{proof}

The fact that the Szeg\"o representation (\ref{szegotruncada}) has
a finite number of terms plays an essential role in the proof of
Theorem \ref{t1}. Although we are mainly interested in evaluating
the integrals $I_{m,n}^{(\lambda)}$ when $\lambda \in\mathbb{N}$,
it is worth pointing out that these integrals can be calculated in
a similar way for all possible values of $\lambda$, provided that
we use the standard trigonometric representation (\ref{standard})
instead of the Szeg\"{o} representation for the Gegenbauer
polynomial inside the logarithm. This generalization is contained
in the next theorem.

\begin{theorem} \label{t2}
For $\lambda \in\mathbb{R}$, $\lambda>-\frac{1}{2}$,
\begin{equation} \label{t2a}
I_{0,n}^{(\lambda)} = 2\pi \log\left( \frac{(\lambda)_n}{n!}\right)
\end{equation}
and, when $m\geq1$,
\begin{equation} \label{formulat2}
I_{m,n}^{(\lambda)}=\frac{\pi}{(2m)!}\,\frac{d^{2m}}{dz^{2m}}
\left. \left( \log \sum_{j=0}^{n}
d_{j,n}^{(\lambda)}z^{2n-2j}\right)\right|_{z=0}\;.
\end{equation}
\end{theorem}

\begin{proof}
We proceed as in the proof of Theorem \ref{t1}, but now we use the
complex form of the standard trigonometric representation
(\ref{standard}) for the Gegenbauer polynomial
$C_n^{(\lambda)}(\cos\theta)$ in (\ref{imntrig}). Thus we arrive at
\begin{equation}
I_{m,n}^{(\lambda)} = \frac{\pi}{(2m)!} \frac{d^{2m}}{dz^{2m}}
\left. \left[(z^{4m}+1)\log \left( \sum_{j=0}^{n}
d_{j,n}^{(\lambda)}z^{2n-2j} \right) \right]\right|_{z=0}\,,
\end{equation}
from which (\ref{t2a}) and (\ref{formulat2}) readily follow. \qed
\end{proof}

In order to carry out the sums in (\ref{suma}), the next step is to
obtain closed formulas for the derivatives $I_{m,n}^{(\lambda)}$
with $1 \leq m \leq n+\lambda$. Despite its greater generality,
Theorem \ref{t2} turns out to be less useful than Theorem \ref{t1},
because (\ref{formulat1}) expresses the integrals in terms of the
logarithm of a polynomial that has $2\lambda$ terms, while in
(\ref{formulat2}) they are given in terms of the logarithm of a
polynomial with $n+1$ terms. As we shall see, the difficulty in
obtaining a closed formula for the derivatives of such functions
increases with the number of terms in the polynomial. Therefore, if
we want an expression of $I_{m,n}^{(\lambda)}$ for a fixed value of
$\lambda$ and any $n \in \mathbb{N}$ Theorem \ref{t1} is more
helpful, particularly for small values of $\lambda$.

In the case $\lambda=1$ we readily notice from (\ref{formulat1})
that, if $1\leq m\leq n+1$, then
\begin{eqnarray}
I_{m,n}^{(1)} & =
\frac{\pi}{m}+\frac{\pi}{(2m)!}\left.\frac{d^{2m}}{dz^{2m}}
\left(\log(1-z^{2n+2})\right)\right|_{z=0} \nonumber \\ \label{i1} &
= \frac{\pi}{m}+\frac{\pi}{(2m)!}\left.\frac{d^{2m}}{dz^{2m}}
\left(-\sum_{k=1}^\infty\frac{z^{(2n+2)k}}{k}\right)\right|_{z=0}
\nonumber \\ & = \pi\left(\frac{1}{m}-\delta_{m,n+1}\right).
\end{eqnarray}
When $\lambda\geq2$, the polynomial inside the logarithm has more
terms and the above trick does not work. However, we can obtain
closed formulas for the derivatives in (\ref{formulat1}) by means
of Fa\`a di Bruno's formula for the derivatives of the composition
of two functions, which states that (see e.g.\ \cite{Roman})
\begin{equation}\label{faa}
\frac{d^mf\big(g(z)\big)}{dz^m} = m! \sum_{k=0}^m
f^{(k)}\big(g(z)\big) \sum_{k_1,k_2,\ldots,k_m}
\prod_{j=1}^m\frac{[g^{(j)}(z)]^{k_j}}{(j!)^{k_j}k_j!}\;,
\end{equation}
where the inner summation is extended over all partitions
satisfying
\begin{equation}\label{particion}
k_1+k_2+\cdots+k_m=k\,,\quad k_1+2k_2+\cdots+mk_m=m\,.
\end{equation}
This formula enables us to find explicit expressions for
$I_{m,n}^{(\lambda)}$ with $\lambda \geq 2$, as stated in the
following two propositions.

\setcounter{theorem}{1}

\begin{proposition}\label{prop3}
In the case $\lambda=2$, when $1\leq m\leq n+2$
\begin{equation}\label{i2}
I_{m,n}^{(2)}=\frac{\pi}{m}\left[ \,
3-\left(\frac{n+3}{n+1}\right)^{\! \! m} \, \right] +
\pi\frac{n+3}{n+1}\,\delta_{m,n+2}\;.
\end{equation}
\end{proposition}

\begin{proof}
In this case, application of Fa\`{a} di Bruno's formula (\ref{faa})
to the derivatives in (\ref{formulat1}) gives\footnote{Notice that
attending to (\ref{particion}) $k=0$ corresponds to $m=0$, so we can
start the sum in $k$ from 1.}
\begin{eqnarray} \label{faa2}
\fl \frac{d^{2m}}{dz^{2m}} \left. \left( \log \sum_{\nu=0}^{1}
\alpha_{\nu,n}^{(2)}(z^{2n+4+2\nu}-z^{2-2\nu}) \right) \right|_{z=0}
= (2m)! \sum_{k=1}^{2m} \left. \frac{d^k}{dz^k} (\log
z)\right|_{z=-\alpha_{1,n}^{(2)}} \nonumber \\ \times
\sum_{k_1,k_2,\ldots,k_{2m}} \prod_{j=1}^{2m}\frac{\left[ \left.
\frac{d^j}{dz^j} \, \big( \sum_{\nu=0}^{1} \alpha_{\nu,n}^{(2)}
(z^{2n+4+2\nu}-z^{2-2\nu}) \big) \right|_{z=0}
\right]^{k_j}}{(j!)^{k_j}k_j!}\;.
\end{eqnarray}
On the one hand, for $k\geq1$,
\begin{equation} \label{dklog}
\frac{d^k}{dz^k}(\log z)=\frac{(-1)^{k+1}(k-1)!}{z^k}\;,
\end{equation}
so that
\begin{equation}
\left. \frac{d^k}{dz^k} (\log z) \right|_{z=-\alpha_{1,n}^{(2)}} =
- \frac{(k-1)!}{(\alpha_{1,n}^{(2)})^k} \;.
\end{equation}
On the other hand, all the derivatives of the polynomial in
(\ref{faa2}) vanish at $z=0$ except when $j=2$ and $j=2n+4$, so we
must set $k_j=0$ if $j\neq2$ and $j\neq2n+4$. Conditions
(\ref{particion}) then read
\begin{equation}
k_2+k_{2n+4}=k\,, \quad 2k_2+(2n+4)k_{2n+4}=2m\,.
\end{equation}
Since $k_2$ and $k_{2n+4}$ are non-negative integers, these
equations only admit the solution $k_{2n+4}=0,k_2=k=m$ when $m\leq
n+1$, while in the case $m=n+2$ we have to add the solution
$k_{2n+4}=k=1,k_2=0$ to the previous one. Therefore, (\ref{faa2})
simplifies to
\begin{equation}
\fl \frac{d^{2m}}{dz^{2m}} \left. \left( \log \sum_{\nu=0}^{1}
\alpha_{\nu,n}^{(2)}(z^{2n+4+2\nu}-z^{2-2\nu}) \right) \right|_{z=0}
= -(2m)! \left[ \, \frac{1}{m} \left( -
\frac{\alpha_{0,n}^{(2)}}{\alpha_{1,n}^{(2)}} \right)^{\! \! m} +
\frac{\alpha_{0,n}^{(2)}}{\alpha_{1,n}^{(2)}} \,\delta_{m,n+2} \,
\right],
\end{equation}
and the result follows using the second equation in (\ref{calfa}).
\qed
\end{proof}

\begin{proposition}\label{prop4}
For any $\lambda\in\N$, $\lambda\geq3$, when $1\leq m\leq n+\lambda$
\begin{eqnarray} \label{lambdageq4}\nonumber
I_{m,n}^{(\lambda)} & = & \frac{(2\lambda-1)\pi}{m} -
\pi\frac{\alpha_{0,n}^{(\lambda)}}{\alpha_{\lambda-1,n}^{(\lambda)}}\;
\delta_{m,n+\lambda} \\ \nonumber & & - \pi
\left(\frac{\alpha_{\lambda-3,n}^{(\lambda)}}{\alpha_{\lambda-2,n}^{(\lambda)}}
\right)^{\!\!m}\sum_{k=1}^m\sum_{k_6=0}^{[m/3]}\sum_{k_8=0}^{[m/4]}\cdots
\sum_{k_{2(\lambda-1)=0}}^{[m/(\lambda-1)]}
\left(-\frac{(\alpha_{\lambda-2,n}^{(\lambda)})^2}
{\alpha_{\lambda-1,n}^{(\lambda)}\alpha_{\lambda-3,n}^{(\lambda)}}
\right)^{\!\!k}\\\nonumber & & \times
\frac{(k-1)!}{\big(2k-m+\sum_{r=1}^{\lambda-3}rk_{2r+4}\big)!\,
\big(m-k-\sum_{s=1}^{\lambda-3}(s+1)k_{2s+4}\big)!}\\ & & \times
\prod_{j=3}^{\lambda-1}\frac{1}{(k_{2j})!}
\left(-\frac{\alpha_{\lambda-1-j,n}^{(\lambda)}
(\alpha_{\lambda-2,n}^{(\lambda)})^{j-2}}{(\alpha_{\lambda-3,n}^{(\lambda)})^{j-1}}
\right)^{\!\!k_{2j}}\,,
\end{eqnarray}
where in the upper limits of the summations over $k_6$, $k_8$,
\ldots, $k_{2(\lambda-1)}$ the square brackets denote integer part
of the expression within.

In particular, in the case $\lambda = 3$, when $1\leq m\leq n+3$
\begin{equation}\label{i3}
I_{m,n}^{(3)} = \frac{\pi}{m} \bigg[ \, 5 - 2 \Re \big( f(n)^m
\big) \, \bigg] - \pi\frac{(n+4)(n+5)}{(n+1)(n+2)} \,
\delta_{m,n+3} \,,
\end{equation}
where
\begin{equation} \label{efedeene}
f(n)=\frac{(n+1)(n+5)+\rmi\sqrt{3(n+1)(n+5)}}{(n+1)(n+2)}\;.
\end{equation}
\end{proposition}

\begin{proof}
In the general case ($\lambda \in \mathbb{N}$, $\lambda \geq 3$),
application of Fa\`{a} di Bruno's formula (\ref{faa}) to the
derivatives in (\ref{formulat1}) and use of (\ref{dklog}) lead to
\begin{eqnarray} \label{dergenlambda}
\fl \frac{d^{2m}}{dz^{2m}} \left. \left( \log
\sum_{\nu=0}^{\lambda-1} \alpha_{\nu,n}^{(\lambda)}
(z^{2n+2\lambda+2\nu}-z^{2\lambda-2\nu-2})\right)\right|_{z=0}
\nonumber \\ \fl = -(2m)! \sum_{k=1}^{2m}
\frac{(k-1)!}{(\alpha_{\lambda-1,n}^{(\lambda)})^k}
\sum_{k_1,k_2,\ldots,k_{2m}} \prod_{j=1}^{2m}
\frac{\left[\left.\frac{d^j}{dz^j} \big(
\sum_{\nu=0}^{\lambda-1}\alpha_{\nu,n}^{(\lambda)}\,z^{2n+2\lambda+2\nu}
\big) \right|_{z=0}\right]^{k_j}}{(j!)^{k_j}k_j!} \nonumber \\ \fl -
(2m)!
\sum_{k=1}^{2m}\frac{(k-1)!}{(\alpha_{\lambda-1,n}^{(\lambda)})^k}
\sum_{k_1,k_2,\ldots,k_{2m}} \prod_{j=1}^{2m}
\frac{\left[\left.\frac{d^j}{dz^j} \big(
-\sum_{\nu=0}^{\lambda-1}\alpha_{\nu,n}^{(\lambda)}
\,z^{2\lambda-2\nu-2} \big)
\right|_{z=0}\right]^{k_j}}{(j!)^{k_j}k_j!}\;.
\end{eqnarray}
In the first term of the right-hand side all derivatives vanish at
$z=0$ except when $j=2n+2\lambda$, so that $k_j=0$ whenever
$j\neq2n+2\lambda$ and conditions (\ref{particion}) simplify to
\begin{equation}
k_{2n+2\lambda}=k\,,\quad (2n+2\lambda)k_{2n+2\lambda}=2m\,,
\end{equation}
which only admit the solution $k_{2n+2\lambda}=k=1$ when
$m=n+\lambda$. In the second term the derivatives that do not vanish
are those with $j$ even, $2\leq j\leq 2\lambda-2$, so that
conditions (\ref{particion}) now read
\begin{equation}\label{particionlambda}
\sum_{r=1}^{\lambda-1}k_{2r}=k\,,\quad\sum_{s=1}^{\lambda-1}sk_{2s}=m\,.
\end{equation}
Equation (\ref{dergenlambda}) thus reduces to
\begin{eqnarray}
\fl \frac{d^{2m}}{dz^{2m}} \left. \left( \log
\sum_{\nu=0}^{\lambda-1} \alpha_{\nu,n}^{(\lambda)}
(z^{2n+2\lambda+2\nu}-z^{2\lambda-2\nu-2})\right) \right|_{z=0} =
-(2n+2\lambda)!
\frac{\alpha_{0,n}^{(\lambda)}}{\alpha_{\lambda-1,n}^{(\lambda)}} \,
\delta_{m,n+\lambda} \nonumber \\ -(2m)! \sum_{k=1}^{2m}
\frac{(k-1)!}{(\alpha_{\lambda-1,n}^{(\lambda)})^k}
\sum_{k_2,k_4,\ldots,k_{2(\lambda-1)}} \prod_{j=1}^{2m}
\frac{(-\alpha_{\lambda-1-j,n}^{(\lambda)})^{k_{2j}}}{(k_{2j})!}\;.
\end{eqnarray}
Finally, we can further simplify the previous expression to obtain
(\ref{lambdageq4}) by using conditions (\ref{particionlambda}) to
write $k_2$ and $k_4$ in terms of the remaining indices, i.e.
\begin{eqnarray}
& k_2=k_6+2k_8+\cdots+(\lambda-3)k_{2(\lambda-1)}+2k-m\,,\nonumber\\
& k_4=-2k_6-3k_8-\cdots-(\lambda-2)k_{2(\lambda-1)}+m-k\,.
\end{eqnarray}
Notice that in (\ref{lambdageq4}) conditions (\ref{particionlambda})
are guaranteed to hold because for the values of the indices that do
not fulfill them we get the inverse of the factorial of a negative
integer, which can be considered to be zero. We have changed the
upper limit in the sum over $k$ from $2m$ to $m$ because when
$m+1\leq k\leq2m$ conditions (\ref{particionlambda}) are not
fulfilled.

In the case $\lambda=3$, (\ref{lambdageq4}) reduces to
\begin{equation}\label{isuma}
\fl I_{m,n}^{(3)} = \frac{5\pi}{m} - \pi \left(
\frac{\alpha_{0,n}^{(3)}}{\alpha_{1,n}^{(3)}} \right)^{\! \! m}
\sum_{k=1}^m \frac{(k-1)!}{(2k-m)!(m-k)!}
\left(-\frac{(\alpha_{1,n}^{(3)})^2}{\alpha_{2,n}^{(3)}
\alpha_{0,n}^{(3)}}\right)^{\! \! k} -
\pi\frac{\alpha_{0,n}^{(3)}}{\alpha_{2,n}^{(3)}}\,\delta_{m,n+3}\;,
\end{equation}
so we need to evaluate a sum of the form
\begin{equation}
\fl \sum_{k=1}^{m}\frac{(k-1)!}{(2k-m)!(m-k)!}x^k =
\sum_{j=0}^{m-1}\frac{(m-j-1)!}{j!(m-2j)!}x^{m-j}=
x^m\sum_{j=0}^{m-1} \left( \begin{array}{c}
m-j \\
j \\
\end{array} \right) \frac{(x^{-1})^{j}}{m-j}\;.
\end{equation}
Using the summation formula \cite[Eq.\ (5.75)]{concrete}
\begin{equation}
\sum_{j=0}^{m-1}
\left(\begin{array}{c}
m-j \\
j \\
\end{array}
\right)\frac{m}{m-j}\,z^j =
\left(\frac{1+\sqrt{1+4z}}{2}\,\right)^{\! \! m} +
\left(\frac{1-\sqrt{1+4z}}{2}\,\right)^{\! \! m}
\end{equation}
and the second equation in (\ref{calfa}), we find that
\begin{equation}
I_{m,n}^{(3)} = \frac{5\pi}{m} - \frac{\pi}{m} \left[
f(n)^m+f^*(n)^m \right] - \pi \frac{(n+4)(n+5)}{(n+1)(n+2)} \,
\delta_{m,n+3} \;,
\end{equation}
which is equivalent to (\ref{i3}). \qed
\end{proof}

\section{Results for the information entropy} \label{section4}

Equations (\ref{icero}), (\ref{i1}), (\ref{i2}) and (\ref{i3})
enable us to derive closed analytical formulas for
$E(C_n^{(\lambda)})$ when $\lambda=1,2,3$. For $\lambda=1$, after
substitution of the corresponding values of the constants
$d_{m,n}^{(\lambda)}$, $c_n^{(\lambda)}$ and
$\alpha_{\nu,n}^{(\lambda)}$ (see (\ref{d}) and (\ref{calfa})),
(\ref{suma}) reduces to
\begin{equation}\label{ec1}
E(C_n^{(1)})=-\frac{1}{2} \left( I_{0,n}^{(1)} - I_{n+1,n}^{(1)}
\right),
\end{equation}
which using (\ref{icero}) and (\ref{i1}) immediately leads to
\begin{equation}\label{gegentlambda1}
E(C_n^{(1)}) = E(U_n) = \frac{\pi}{2}\left(\frac{1}{n+1}-1\right).
\end{equation}
When $\lambda=2$, (\ref{suma}) takes the form
\begin{equation}
E(C_n^{(2)})=-\frac{1}{8} \left( (n+1)(n+3) I_{0,n}^{(2)}- (n+1)^2
I_{n+2,n}^{(2)} - 4\sum_{m=1}^{n+1} m I_{m,n}^{(2)} \right),
\end{equation}
so using (\ref{icero}) and (\ref{i2}) together with the well-known
formula for the sum of a geometric series,
\begin{equation} \label{geosum}
\sum_{m=1}^n x^m = \frac{x(1-x^n)}{1-x}\,,
\end{equation}
we find that
\begin{eqnarray} \label{gegentlambda2}
\fl E(C_n^{(2)}) = -\frac{\pi}{8} \left( 2(n+1)(n+3)\log(n+1) +
\frac{n^3-5n^2-29n-27}{n+2} \right. \nonumber \\ \left. +
\frac{(n+3)^{n+3}}{(n+2)(n+1)^{n+1}} \right).
\end{eqnarray}
Recalling (\ref{eogeonrel}), (\ref{gegentlambda1}) and
(\ref{gegentlambda2}) are readily shown to be equivalent to
(\ref{echebyshev2}) and (\ref{gegl2entropy}), respectively.

In the case $\lambda=3$, (\ref{suma}) can be rewritten as
\begin{eqnarray}\nonumber
\fl E(C_n^{(3)}) = -\frac{1}{128} \left[ (n+1)(n+2)(n+4)(n+5)
I_{0,n}^{(3)} \right. \\
\left. -(n+1)^2(n+2)^2 I_{n+3,n}^{(3)} -12 \sum_{m=1}^{n+2}
m(n^2+6n+7-2m^2) I_{m,n}^{(3)}\right]\,.
\end{eqnarray}
Substituting (\ref{icero}) and (\ref{i3}) into the above expression,
we encounter again the geometric sum, as well as a sum of the form
$\sum_mm^2x^m$. Using (\ref{geosum}) and the summation formula
\cite[Eq.\ (5.14.9)]{Han}
\begin{equation}
\sum_{m=1}^n m^2 x^m = \frac{x(1+x) -
x^{n+1}\big[(n+1)^2-(2n^2+2n-1)x+n^2x^2\big]}{(1-x)^{3}}\,,
\end{equation}
after a tedious but straightforward calculation we arrive at the
following closed analytical formula for $E(C_n^{(3)})$, which is a
new result:
\begin{eqnarray}\nonumber
\fl E(C_n^{(3)}) & = & -\frac{\pi}{128} \left\{
2(n+1)(n+2)(n+4)(n+5)\log\left(\frac{(n+1)(n+2)}{2}\right) \right.
\\\nonumber \fl & & +
\frac{n^5-16n^4-269n^3-1200n^2-2102n-1250}{n+3} \\\nonumber
        \fl & & + \frac{2(n+5)^2}{(n+2)(n+3)}\,\Re\left[ \left(
        \frac{(n+1)(n+5)+\rmi\sqrt{3(n+1)(n+5)}}{(n+1)(n+2)}\;
        \right)^{\!\! n+1}\right. \\\label{ec3}
        \fl & & \left.\left. \times
        \left(2n^2+13n+14- \rmi (n+1)(n+6)\sqrt{\frac{(n+1)(n+5)}{3}}\;\;\right)
        \right] \right\}.
\end{eqnarray}

When $\lambda \geq 4$, combination of (\ref{suma}) and
(\ref{lambdageq4}) provide an expression for the entropy
$E(C_n^{(\lambda )})$ in terms of finite sums. For the sake of
brevity, in (\ref{suma}) it is convenient to absorb the term
corresponding to $I_{n+\lambda,n}^{(\lambda)}$ into the sum over $m$
by setting
\begin{equation}
\beta_{n+\lambda,n}^{(\lambda)} :=
-\alpha_{\lambda-1,n}^{(\lambda)}d_{n,n}^{(\lambda)}
\end{equation}
instead of using (\ref{beta}) with $m=n+\lambda$, which would give
for $\beta_{n+\lambda,n}^{(\lambda)}$ the value
$-\alpha_{\lambda-1,n}^{(\lambda)}d_{n,n}^{(\lambda)} +
\alpha_{0,n}^{(\lambda)}d_{n+\lambda,n}^{(\lambda)}$. We thus have
that
\begin{eqnarray}
\fl E(C_n^{(\lambda)}) & = & -\frac{\pi}{2}\,
c_n^{(\lambda)}\left\{2\alpha_{0,n}^{(\lambda)}d_{0,n}^{(\lambda)}
\log\left( \frac{(\lambda)_n}{n!}\right)
+\alpha_{0,n}^{(\lambda)}d_{n,n}^{(\lambda)} +(2\lambda-1)
\sum_{m=1}^{n+\lambda}\frac{\beta_{m,n}^{(\lambda)}}{m}\right.
\nonumber\\ \fl & & - \sum_{m=1}^{n+\lambda}
\sum_{k=1}^m\sum_{k_6=0}^{[m/3]}\sum_{k_8=0}^{[m/4]}\cdots
\sum_{k_{2(\lambda-1)=0}}^{[m/(\lambda-1)]}\beta_{m,n}^{(\lambda)}
\left(
\frac{\alpha_{\lambda-3,n}^{(\lambda)}}{\alpha_{\lambda-2,n}^{(\lambda)}}
\right)^{\!\!m} \left(
-\frac{(\alpha_{\lambda-2,n}^{(\lambda)})^2}{\alpha_{\lambda-1,n}^{(\lambda)}
\alpha_{\lambda-3,n}^{(\lambda)}}
\right)^{\!\!k}\nonumber\\\nonumber \fl & & \times
\frac{(k-1)!}{\big(2k-m+\sum_{r=1}^{\lambda-3}rk_{2r+4}\big)!\,
\big(m-k-\sum_{s=1}^{\lambda-3}(s+1)k_{2s+4}\big)!}\\
\fl & & \left. \times \prod_{j=3}^{\lambda-1} \frac{1}{(k_{2j})!}
\left( -\frac{\alpha_{\lambda-1-j,n}^{(\lambda)}
(\alpha_{\lambda-2,n}^{(\lambda)})^{j-2}}{
(\alpha_{\lambda-3,n}^{(\lambda)})^{j-1}}
\right)^{\!\!k_{2j}}\right\}\,.\label{lambdamayor4}
\end{eqnarray}
Unlike (\ref{buyarov}), (\ref{lambdamayor4}) is completely
analytical for all $\lambda \in \mathbb{N}$, which makes it suitable
for symbolic computation. For instance, a Maple implementation of
the formula enabled us to obtain the closed analytical expressions
for $E(C_n^{(4)})$ and $E(C_n^{(5)})$, with $1 \leq n \leq 15$, that
are displayed in Tables \ref{table1} and \ref{table2}, respectively.
In these tables we also provide numerical values of the entropies
obtained from the exact ones, in order that the interested reader
can compare them with those given by numerical algorithms such as
that in \cite{Buy04}.

\begin{table}
\caption{Exact and numerical values of the entropy $E(C_n^{(4)})$
for $1 \leq n \leq 15$.}
\centering
{\footnotesize
\begin{tabular}{|r|r|r|}
\hline {n} & {Exact value} & {Numerical value} \\\hline
    &   &   \\
1   & $\displaystyle{-7\pi\log(2) + \frac{119}{240}\pi}$ & $-13.685$ \\
    &   &   \\
2   & $\displaystyle{-\frac{105}{8}\pi\log(10)
+\frac{580771}{300000}\pi}$ & $-88.862$ \\
    &   &   \\
3   & $\displaystyle{-\frac{75}{2}\pi\log(20) + \frac{95}{21}\pi}$
& $-338.714$ \\
    &   &   \\
4   & $\displaystyle{-\frac{5775}{64}\pi\log(35) +
\frac{4883222845}{632481024}\pi}$ & $-983.613$\\
    &   &   \\
5   & $\displaystyle{-\frac{385}{2}\pi\log(56) +
\frac{17355685}{1806336}\pi}$ & $-2404.173$ \\
    &   &   \\
6   & $\displaystyle{-{\frac {3003}{8}}\,\pi\,\log  \left( 84
\right)
+ {\frac {6449434961}{1058158080}}\,\pi}$ & $- 5206.005$ \\
    &   &   \\
7   & $\displaystyle{-{\frac {1365}{2}}\,\pi\,\log  \left( 120
\right) -{\frac {1396715852287}{139218750000}}\,\pi}$ & $-10296.556$ \\
    &   &   \\
8   & $\displaystyle{-{\frac {75075}{64}}\,\pi\,\log  \left( 165
\right) -{\frac {24757176334716125}{493018566815808}}\,\pi}$ &
$-18974.368$ \\
    &   &   \\
9   & $\displaystyle{ -1925\,\pi\,\log  \left( 220 \right) -{\frac
{1200329915}{9135984}}\,\pi}$ & $-33031.075$ \\
    &   &   \\
10  & $\displaystyle{ -{\frac {12155}{4}}\,\pi\,\log  \left( 286
\right) -{\frac
{325291539600149215255}{1172732412725203616}}\,\pi}$ &
$-54866.421$ \\
    &   &   \\
11  & $\displaystyle{ -4641\,\pi\,\log  \left( 364 \right) -{\frac
{31458443588344487293819}{60436675052957701680}}\,\pi}$ & $-87616.538$ \\
    &   &   \\
12  & $\displaystyle{ -{\frac {440895}{64}}\,\pi\,\log  \left( 455
\right) -{\frac
{25537984326378849719971131}{28270687046875000000000}}\,\pi}$ & $-135295.739$ \\
    &   &   \\
13  & $\displaystyle{ -9975\,\pi\,\log  \left( 560 \right) -{\frac
{1779685691911133495}{1202109806542848}}\,\pi}$ & $-202952.031$ \\
    &   &   \\
14  & $\displaystyle{ -{\frac {56525}{4}}\,\pi\,\log  \left( 680
\right) -{\frac
{36234350694889865223938313068785}{15613637127259094259005915136}}\,\pi}$
& $-296836.555$ \\
    &   &   \\
15  & $\displaystyle{ -19635\,\pi\,\log  \left( 816 \right)
-{\frac
{130243656594168370141034405}{37115886521993021558784}}\,\pi}$ & $-424587.139$ \\
    &   &   \\\hline
\end{tabular}
\label{table1} } \vspace*{-13pt}
\end{table}

\begin{table}
\caption{Exact and numerical values of the entropy $E(C_n^{(5)})$
for $1 \leq n \leq 15$.}
\centering
{\footnotesize
\begin{tabular}{|r|r|r|}
\hline {n} & {Exact value} & {Numerical value} \\\hline
    &   &   \\
1   & $\displaystyle{ -{\frac {525}{128}}\,\pi\,\log  \left( 5
\right) +{\frac {945}{1024}}\,\pi}$
&$-17.839$ \\
    &   &   \\
2   & $\displaystyle{ -{\frac {2475}{128}}\,\pi\,\log  \left( 15
\right) +{\frac
{27685925}{5225472}}\,\pi}$ &$-147.857$\\
    &   &   \\
3   & $\displaystyle{ -{\frac {17325}{256}}\,\pi\,\log  \left( 35
\right) +{\frac
{61634724075}{3373232128}}\,\pi}$ &$-698.499$ \\
    &   &   \\
4   & $\displaystyle{ -{\frac {25025}{128}}\,\pi\,\log  \left( 70
\right) +{\frac
{5573831525}{115605504}}\,\pi}$ &$-2457.981$ \\
    &   &   \\
5   & $\displaystyle{ -{\frac {63063}{128}}\,\pi\,\log  \left( 126
\right) +{\frac
{338107973281463}{3173748645888}}\,\pi}$ &$-7150.909$ \\
    &   &   \\
6   & $\displaystyle{ -{\frac {143325}{128}}\,\pi\,\log  \left(
210 \right) +{\frac
{20887195}{101376}}\,\pi}$ &$-18162.369$ \\
    &   &   \\
7   & $\displaystyle{ -{\frac {75075}{32}}\,\pi\,\log  \left( 330
\right) +{\frac
{1408430247274269205}{3944148534526464}}\,\pi}$ &$-41620.201$ \\
    &   &   \\
8   & $\displaystyle{ -{\frac {294525}{64}}\,\pi\,\log  \left( 495
\right) +{\frac
{806559968327725}{1438588584576}}\,\pi}$ &$-87940.792$ \\
    &   &   \\
9   & $\displaystyle{ -{\frac {546975}{64}}\,\pi\,\log  \left( 715
\right) +{\frac
{29915266041851863399425}{37527437207206515712}}\,\pi}$ &$-173958.634$ \\
    &   &   \\
10  & $\displaystyle{ -{\frac {969969}{64}}\,\pi\,\log  \left(
1001 \right) +{\frac
{97664804776687286561309}{96698680084732322688}}\,\pi}$ &$-325775.232$ \\
    &   &   \\
11  & $\displaystyle{ -{\frac {6613425}{256}}\,\pi\,\log  \left(
1365 \right) +{\frac
{230209361727271224010045}{212679240849405517824}}\,\pi} $ &$-582478.486$ \\
    &   &   \\
12  & $\displaystyle{ -{\frac {2723175}{64}}\,\pi\,\log  \left(
1820 \right) +{\frac
{10188450005911283085}{12635587626401792}}\,\pi}$ &$-1000899.539$ \\
    &   &   \\
13  & $\displaystyle{ -{\frac {4352425}{64}}\,\pi\,\log  \left(
2380 \right) -{\frac
{39663465263970548089202600252605}{249818194036145508144094642176}}\,\pi}$
&$-1661590.212$ \\
    &   &   \\
14  & $\displaystyle{ -{\frac {6774075}{64}}\,\pi\,\log  \left(
3060 \right) -{\frac
{717231543067734581588054629334401175}{306761704661739640893688309874688}}\,\pi}$
&$-2676220.464$ \\
    &   &   \\
15  & $\displaystyle{ -{\frac {5148297}{32}}\,\pi\,\log  \left(
3876 \right) -{\frac
{535111116210266542852402527915814650511}{82233794352493419438828330115762176}}\,\pi}$
&$-4196611.889$ \\
    &   &   \\\hline
\end{tabular}
\label{table2} } \vspace*{-13pt}
\end{table}

\section{Summary and conclusions} \label{section5}

The problem of obtaining closed analytical formulas for the
entropy of orthogonal polynomials is known to be very difficult,
as displayed by the fact that in previous work on the subject
formulas of this kind were only found for the Gegenbauer
polynomials of parameter $\lambda = 0, 1, 2$. Here we have
presented a new approach to the calculation of the entropy of
Gegenbauer polynomials, based on the use of trigonometric
representations for these polynomials, which has allowed us to
explicitly evaluate the entropic integrals by means of complex
analysis techniques. Using this method we have been able to derive
in a unified way closed formulas of $E(C_n^{(\lambda)})$ for
$\lambda=1, 2, 3$, the last one being new. Furthermore, when
$\lambda\geq4$, $\lambda \in \mathbb{N}$, we have obtained
completely analytical expressions of the entropy in terms of
finite sums, which easily provide exact values for the entropy
using symbolic computation. The growing complexity in the formulas
of $E(C_n^{(\lambda )})$ as $\lambda$ increases serves as a clear
illustration of the difficulties posed by the calculation of the
entropy of orthogonal polynomials.

When the parameter $\lambda$ is not a positive integer, the
Szeg\"{o} representation (\ref{szego}) of the Gegenbauer polynomial
$C_n^{(\lambda)}$ has infinitely many terms, so the same happens for
the expressions (\ref{entrotrigo2}) and (\ref{entrotrigo3}) of the
entropic integral $E(C_n^{(\lambda)})$. It remains open the problem
of studying the convergence behaviour of these series, as well as
that of summing up them analytically. It would be of particular
interest to obtain exact analytical expressions for the entropy of
Gegenbauer polynomials of half-integer parameter since, as already
mentioned in Sec.\ \ref{section1}, they are needed together with
those of the integer case in order to evaluate the information
entropy of spherical and hyperspherical harmonics. Finally, it would
also be desirable to extend the method introduced in this paper to
other families of orthogonal polynomials having trigonometric
representations, a line of research that is currently being
developed.

\begin{ack}

The second author (S. Gandy) gratefully acknowledges the hospitality
of the Departamento de Matem\'aticas of the Universidad Carlos III
de Madrid, where this research was carried out, as well as financial
support from the European Union Socrates/Erasmus Programme. The work
of the first and third authors (J.I. de Vicente and J.
S\'anchez-Ruiz) was supported by Universidad Carlos III de Madrid,
Comunidad Aut\'onoma de Madrid (project No.\ UC3M-MTM-05-033), and
Direcci\'on General de Investigaci\'on (MEC) of Spain under grant
MTM2006-13000-C03-02. The work of the third author was also
supported by the Direcci\'on General de Investigaci\'on (MEC) of
Spain grant FIS2005-00973, and the Junta de Andaluc\'{\i}a research
group FQM-0207.

\end{ack}


\section*{References}

\begin{thebibliography}{10}

\bibitem{Sha} Shannon C E 1948 A mathematical theory of
communication \emph{Bell Syst. Tech. J.} {\bf 27} 379--423, 623--56
Reprinted in Shannon C E and Weaver W (ed) 1949 \emph{The
Mathematical Theory of Communication} (Urbana, IL: University of
Illinois Press)

\bibitem{Bec} Beckner W 1975 Inequalities in Fourier analysis
\emph{Ann. Math.} {\bf 102} 159--82

\bibitem{Bia} Bialynicki-Birula I and Mycielski J 1975 Uncertainty
relations for information entropy in wave mechanics \emph{Commun.
Math. Phys.} {\bf 44} 129--32

\bibitem{review} Dehesa J S, Mart\'{\i}nez-Finkelshtein A and
S\'anchez-Ruiz J 2001 Quantum information entropies and orthogonal
polynomials \emph{J. Comput. Appl. Math.} {\bf 133} 23--46

\bibitem{Buy04} Buyarov V S, Dehesa J S, Mart\'inez-Finkelshtein A and
S\'anchez-Lara J 2004 Computation of the entropy of polynomials
orthogonal on an interval \emph{SIAM J. Sci. Comput.} {\bf 26}
488--509

\bibitem{Sze} Szeg\H{o} G 1975 \emph{Orthogonal Polynomials (Am. Math.
Soc. Colloq. Publ. vol 23)} 4th edn (Providence, RI: American
Mathematical Society)

\bibitem{Yan94} Y\'{a}\~{n}ez R J, Van Assche W and Dehesa J S 1994
Position and momentum information entropies of the $D$-dimensional
harmonic oscillator and hydrogen atom \emph{Phys. Rev. A} {\bf 50}
3065--79

\bibitem{Deh97} Dehesa J S, Van Assche W and Y\'{a}\~{n}ez R J 1997
Information entropy of classical orthogonal polynomials and their
application to the harmonic oscillator and Coulomb potentials
\emph{Meth. Appl. Anal.} {\bf 4} 91--110

\bibitem{Yan99} Y\'{a}\~{n}ez R J, Van Assche W,
Gonz\'{a}lez-F\'{e}rez R and Dehesa J S 1999 Entropic integrals of
hyperspherical harmonics and spatial entropy of $D$-dimensional
central potentials \emph{J. Math. Phys.} {\bf 40} 5675--86

\bibitem{Buyarov:97} Buyarov V S 1997 On information entropy of Gegenbauer
polynomials \emph{Vestn.\ Mosk.\ Univ. Ser.\ 1 Mat. Mekh.} {\bf 6}
8--11 (in Russian)

\bibitem{Buy00} Buyarov V S, L\'opez-Art\'es P,
Mart\'{\i}nez-Finkelshtein A and Van Assche W 2000 Information
entropy of Gegenbauer polynomials \emph{J. Phys. A: Math. Gen.} {\bf
33} 6549--60

\bibitem{San03} S\'anchez-Ruiz J 2003 Information entropy of
Gegenbauer polynomials and Gaussian quadrature \emph{J. Phys. A:
Math. Gen.} {\bf 36} 4857--66

\bibitem{Sanpreprint} S\'anchez-Lara J and S\'anchez-Ruiz J 2005
Information entropy of the Jacobi polynomials $P_n^{(m-1/2,\pm1/2)}$
\emph{Preprint} (unpublished)

\bibitem{AAR} Andrews G E, Askey R and Roy R 1999 \emph{Special
Functions} (Cambridge: Cambridge University Press)

\bibitem{Sze2} Szeg\H{o} G 1934 \"{U}ber gewisse orthogonale
polynome, die zu einer oszillierenden Belegungsfunktion geh\"{o}ren
\emph{Math. Ann.} {\bf 110} 501--13 (in German)

\bibitem{Roman} Roman S 1980 Fa\`{a} di Bruno's formula
\emph{Amer. Math. Monthly} {\bf 87} 805--9

\bibitem{concrete} Graham R L, Knuth D E and Patashnik O 1994
\emph{Concrete Mathematics} (Reading, MA: Addison-Wesley)

\bibitem{Han} Hansen E R 1975 \emph{A Table of Series and
Products} (Englewood Cliffs, NJ: Prentice-Hall)

\end{thebibliography}
\end{document}